\documentclass{article}
\usepackage[utf8]{inputenc}
\usepackage{physics}
\usepackage{amsmath,graphicx}
\usepackage{mathtools}
\mathtoolsset{showonlyrefs=true}
\usepackage{times}
\usepackage{amsmath}
\usepackage{amsfonts}
\usepackage{amsthm}
\usepackage{enumitem}  
\usepackage{xfrac}
\usepackage[ruled,vlined]{algorithm2e}
\usepackage{hyperref}
 \hypersetup{
     colorlinks=true,
     linkcolor=blue,
     filecolor=blue,
     citecolor = blue,      
     urlcolor=blue,
     }
\usepackage{authblk}

\usepackage{blindtext} 


\newtheorem{theorem}{Theorem}[section]

\newtheorem{lemma}[theorem]{Lemma}
\newtheorem{proposition}{Proposition}[section]
\newtheorem{remark}{Remark}[section]

\newcommand{\inp}[0]{{z}}
\newcommand{\inpS}[0]{{\mathcal{Z}}}

\usepackage{interval}
\intervalconfig{soft open fences}
\usepackage{xcolor}
\usepackage{braket}
\usepackage{changes}
\DeclareMathOperator{\inte}{Int}
\DeclareMathOperator{\dom}{dom}
\DeclareMathOperator{\supp}{supp}

\title{Minimizing Quantum R\'enyi Divergences via Mirror Descent with {Polyak Step Size}}

\author{\(^{1}\) Jun-Kai~You\thanks{ Jun-Kai~You (\href{mailto:tom.junkai.you@gmail.com}{\texttt{tom.junkai.you@gmail.com}}) and Yen-Huan~Li \\ (\href{mailto:yenhuan.li@csie.ntu.edu.tw}{\texttt{yenhuan.li@csie.ntu.edu.tw}}) are supported by the Young Scholar Fellowship (Einstein Program) of the Ministry of Science and Technology of Taiwan under grant number MOST 110-2636-E-002-012.}, \(^{2, 3, 4, 5}\)Hao-Chung~Cheng\thanks{
H.-C.~Cheng (\href{mailto:haochung.ch@gmail.com}{\texttt{haochung.ch@gmail.com}}) is supported by the Young Scholar Fellowship (Einstein Program) of the Ministry of Science and Technology in Taiwan (R.O.C.) under grant number MOST 109-2636-E-002-001 \& 110-2636-E-002-009, and is supported by the Yushan Young Scholar Program of the Ministry of Education in Taiwan (R.O.C.) under grant number NTU-109V0904 \& NTU-110V0904 \& NTU-111L894605.}, \(^{1,2,5}\) Yen-Huan~Li}

\affil{\small $^1$Department of Computer Science and Information Engineering, National Taiwan University \\
$^2$Department of Mathematics, National Taiwan University\\
$^3$Department of Electrical Engineering and Graduate Institute of Communication Engineering, National Taiwan University\\ 
$^4$Hon Hai (Foxconn) Quantum Computing Centre\\
$^5$Center for Quantum Science and Engineering,  National Taiwan University
}

\date{}
\begin{document}

\maketitle

\begin{abstract}
Quantum information quantities play a substantial role in characterizing operational quantities in various quantum information-theoretic problems. 
We consider the numerical computation of four quantum information quantities: Petz-Augustin information, sandwiched Augustin information, conditional sandwiched R\'enyi entropy, and sandwiched R\'enyi information. 
To compute these quantities requires minimizing some order-\( \alpha \) quantum R\'enyi divergences over the set of quantum states. 
Whereas the optimization problems are obviously convex, they violate standard bounded gradient/Hessian conditions in literature, so existing convex optimization methods and their convergence guarantees do not directly apply. 
In this paper, we propose a new class of convex optimization methods called mirror descent with the Polyak step size. 
We prove their convergence under a weak condition, showing that they provably converge for minimizing quantum R\'enyi divergences. 
Numerical experiment results show that entropic mirror descent with the Polyak step size converges fast in minimizing quantum R\'enyi divergences. 
\end{abstract}

\section{Introduction} 
\label{sec:intro}
Quantum entropic and information quantities play a substantial role in characterizing operational quantities in various quantum information-theoretic tasks \cite{Tom16, Wil16, Hay17, KW20}.
For example, the quantum Augustin and R\'enyi information determine the error exponent of communications over a classical-quantum channel using certain random codebooks and that of composite quantum hypothesis testing \cite{DW14, CHT19, CGH18, MO17, MO21, HT14}. 
The quantum conditional R\'enyi entropies have profound applications to quantum cryptographic protocols, such as bounding the convergence rate of privacy amplification \cite{Dup21}, analyzing the security of quantum key distributions \cite{Ren05, Konig_2009, Tom12}, and bounding the error exponent of source coding with quantum side information \cite{CHDH2-2018, Cheng2021a}. 
The above-mentioned quantities share a common feature---they require optimizing some order-$\alpha$ quantum R\'enyi divergence over the set of all quantum states.
Hence, an efficient algorithm for solving such optimization problems is crucial to understanding the fundamental limit of several quantum information-theoretic tasks.

Numerically optimizing quantum R\'enyi divergences is a non-trivial task. 
First, opposite to the classical\footnote{In the classical case, all the underlying density matrices mutually commute, i.e.~they share a common eigenbasis.} scenario where most optimizers have analytic forms, optimizations involving the sandwiched quantum R\'enyi divergences lack closed-form expressions due to the noncommutative nature.
Second, describing a quantum state amounts to determining $4^n-1$ parameters when the underlying Hilbert space is an $n$-qubit system. Hence, na{\"i}ve searching is computationally demanding.

Whereas the problem of optimizing quantum R\'enyi divergences is obviously convex, standard convex optimization algorithms do not directly apply. 
Indeed, even projected gradient descent (PGD), arguably the simplest convex optimization algorithm, is not well-defined when optimizing quantum R\'enyi divergences. 
Notice that projection (with respect to the Frobenius norm) onto the set of quantum states often results in an extreme point, a low-rank density matrix, of the set of quantum states, and the gradients of quantum R\'enyi divergences may not exist at a low-rank density matrix. 
Then, it can happen that PGD is forced to stop at a low-rank iterate far away from the minimizer, because the next iterate is not well-defined.
A similar issue also appears in quantum process tomography and is named the ``stalling issue'' in \cite{Knee2018a}. 
Indeed, we found it difficult to tune the step size of PGD to make it empirically work in minimizing quantum R\'enyi divergences, especially when the number of qubits is large. 

The stalling issue vanishes if we use a more complicated optimization algorithm called entropic mirror descent (aka the exponentiated gradient method) \cite{Beck2003,Li2019a,Youssry2019}, as it always generates full-rank iterates. 
Moreover, entropic mirror descent is known to be scalable with respect to the parameter dimension when the constraint set is the set of quantum states (see, e.g, \cite{BenTal2015}), a desirable feature for quantum information studies. 
Unfortunately, standard entropic mirror descent lacks a convergence guarantee for minimizing quantum R\'enyi divergences. 
Entropic mirror descent is an instance of a class of optimization methods called mirror descent. 
Standard convergence analyses of mirror descent assume the function to be minimized has a bounded gradient or Hessian \cite{Auslender2006,Beck2003}. 
Recently, the bounded Hessian condition is relaxed to the relative smoothness condition \cite{Bauschke2017,Lu2018,Teboulle2018}. 
We show in Proposition \ref{proposition:unbounded_Augustin} and Proposition \ref{proposition:unbounded_conditional} that the bounded gradient/Hessian and relative smoothness conditions do not hold for quantum R\'enyi divergences. 


With regard to the above discussions, in this paper, we propose a variant of mirror descent called mirror descent with the Polyak step size. 
We prove its convergence for convex optimization under a very weak assumption that the gradient of the function to be minimized is locally bounded. 
We highlight the following two contributions. 
\begin{itemize}
  \item In theory, to the best of our knowledge, the proposed method is the first guaranteed-to-converge extension of the Polyak step size to mirror descent. 
  Previously, the Polyak step size was only considered for gradient descent-type methods. 
  
  \item In practice, numerical results show entropic mirror descent with the Polyak step size converges fast for minimizing quantum R\'enyi divergences. 
  Recall that the stalling issue vanishes with entropic mirror descent. 
  Moreover, it can be checked that our convergence guarantee applies to minimizing quantum R\'enyi divergences (see Section \ref{Applicability_of_thm}). 
  
 
\end{itemize}



%

\section{Related Work}
\label{sec:related_work}
We are not aware of any existing guaranteed-to-converge method that can minimize various notions of quantum R\'enyi divergences on the set of quantum states. 
For the specific task of minimizing the Petz--R\'enyi divergence of order \( \alpha \) (required for computing the Petz--Augustin information), a fixed-point iteration is proved to converge for \( \alpha \in \interval[open left]{0}{1} \) in \cite{Nakiboglu2019}\footnote{Via private communication with Nak{\i}bo\u{g}lu, we remark that his approach \cite[Lemma 13]{Nakiboglu2019} can be extended to proving the asymptotic convergence for the range $\alpha \in \interval[open]{1}{2}$.}, \emph{for the commuting case}. 
We show in Section \ref{sec:numerics} that a na\"{\i}ve extension of the fixed-point iteration to the non-commutative case seems to converge, but its convergence speed is significantly slower than entropic mirror descent with the Polyak step size. 

Mirror descent is a class of convex optimization methods that include, for example, gradient descent and entropic mirror descent as special cases (see, e.g., \cite{Beck2003,Nemirovsky1983}). 
Unfortunately, existing convergence analyses of mirror descent do not directly apply to minimizing quantum R\'enyi divergences, which we have already discussed. 
A variant of entropic mirror descent, entropic mirror descent with Armijo line search, is guaranteed to converge for convex differentiable minimization on the set of quantum states and hence applicable to our problem of minimizing quantum R\'enyi divergences \cite{Li2019a}. 
We show in Section \ref{sec:numerics} that entropic mirror descent with Armijo line search is also significantly slower than entropic mirror descent with the Polyak step size. 

The Polyak step size is an adaptive step size selection rule for gradient descent-type methods. 
Polyak proved the convergence of gradient descent with the Polyak step size for bounded-Hessian and strongly convex functions \cite{Polyak1969} (see also \cite{Allen1987,Goffin1977}). 
Nedi\'c and Bertsekas proved the convergence of incremental gradient descent with the Polyak step size in \cite{Nedic2001}, and analyzed the convergence rate in \cite{Nedic2000}, assuming the function to be minimized has a bounded gradient. 
Recently, there have been some works on stochastic gradient descent and mirror descent with the Polyak step size \cite{Berrada2020,Loizou2021,Oberman2019, Orazio2021}. 
These recent works, as well as the original Polyak step size, require the optimal value of the optimization problem. 
This requirement is unrealistic for our problem of minimizing quantum R\'enyi divergences. 
The step size selection rule we propose is inspired by the ``first adjustment'' of the Polyak step size introduced in \cite{Nedic2001}, which replaces the optimal value with a sequence of estimates of the optimal value. 

\section{Problem Formulation}
\label{sec:problem}

We formulate the four information-theoretic quantities we aim to compute as solutions to optimization problems.
Recall that we have addressed the relevance of the quantities in Section \ref{sec:intro}. 
As our algorithm is a first-order algorithm, we also provide gradient formulas for the four optimization problems. 

Throughout this paper, $\mathcal{H}_B := \mathbb{C}^d$ denotes a $d$-dimensional Hilbert space. Let $\mathcal{B}(\mathcal{H}_B)$ be the space of self-adjoint bounded matrices on $\mathcal{H}_B$ endowed with the Hilbert--Schmidt inner product $\langle \cdot, \cdot \rangle$.
For any matrix $\sigma_B \in \mathcal{B}(\mathcal{H}_B)$, we denote by $\supp(\sigma_B):= \{|\psi\rangle\in \mathcal{H}_B \mid \sigma_B |\psi\rangle \neq 0 \}$ the support of $\sigma_B$. We use $\|\sigma_B\|_p := \left( 
\Tr\left[|\sigma_B|^p\right] \right)^{\sfrac{1}{p}}$ to stand for the Schatten $p$-norm.

The gradient formulas in this section are derived as follows. 
Given an extended real-valued function $f:\mathcal{D}(\mathcal{H}_B)\to \mathbb{R}\cup\{\pm\infty\}$, we can represent its Fr\'echet derivative \( \mathsf{D} f \) (see e.g.~\cite[§ X]{Bhatia1997}, \cite[§ 3]{HP14}) at $\sigma \in \mathcal{D}(\mathcal{H}_B)$ as a linear functional by 
\begin{align} \label{eq:def}
\mathsf{D}f(\sigma)[\Delta] = \langle \nabla f(\sigma), \Delta \rangle, \quad \forall \, \text{ traceless }\, \Delta\in \mathcal{B}(\mathcal{H}_B),  
\end{align}
for some unique bounded matrix \( \nabla f ( \sigma ) \). 
We call $\nabla f(\sigma)$ the \emph{gradient of $f$ at $\sigma$}.
We further define 
\[
\dom \nabla f \coloneqq \{\sigma \in \mathcal{D}(\mathcal{H}_B): \nabla f(\sigma) \in \mathcal{B}(\mathcal{H}_B)  \} . 
\]

\subsection{Petz--Augustin and Sandwiched Augustin Information}

Let $\rho_{XB} = \oplus_{x} P_{X}(x)\rho_{B}^{x}$ be a classical-quantum state on the Hilbert space $\mathcal{H}_X\otimes \mathcal{H}_B$, where $\mathcal{H}_X$ can be viewed as a finite set; $P_{X}$ is a probability distribution for the random variable $X$ on it; and each $\rho_B^x$ is a density matrix (i.e.~a positive semi-definite matrix with unit trace) on the $\mathcal{H}_B$.
Without loss of generality, 
we assume that $\cup_{x:P_X(x)>0} \supp\left(\rho_B^x\right) = \mathcal{H}_B$.
We denote by \( \mathcal{D} ( \mathcal{H}_B ) \) the set of quantum density matrices on \( \mathcal{H}_B \), i.e., 
\begin{equation} \label{eq:set_B}
\mathcal{D}(\mathcal{H}_B) = \set{ \sigma_B \in \mathcal{B}(\mathcal{H}_B) | \sigma_B \geq 0, \operatorname{Tr}[\sigma_B] = 1}.
\end{equation}

The order-$\alpha$ \textit{Petz--Augustin Information} for $\rho_{XB}$ is defined by
\begin{equation} \label{eq:Petz_Augustin}
    \min_{\sigma_B \in \mathcal{D}(\mathcal{H}_B)} f_{\mathrm{Petz}}(\sigma_B), \quad f_{\text{Petz}} ( \sigma_B ) \coloneqq \mathbb{E}_{P_X}{D_{\alpha}\left(\rho_B^X \vert\vert\sigma_B\right)},
\end{equation}
where $D_{\alpha}(\rho \vert \vert\sigma)$ for density matrices $\rho,\sigma$ denotes the \emph{Petz--R\'enyi divergence} \cite{Pet86}:
\begin{equation}
    D_{\alpha}( \rho \vert \vert\sigma) := \frac{1}{\alpha-1}\log\operatorname{Tr}\left[ \rho^{\alpha}\sigma^{1-\alpha}\right], \quad \alpha \in \interval[open left]{0}{2} \backslash\{1\},
\end{equation}
when the support of $\rho$ is contained in the support of $\sigma$, and it is defined to be $+\infty$ otherwise.
The matrix function follows from the standard functional calculus; namely, the power function is applied to the eigenvalues of the specified matrix.
The minimizer of \eqref{eq:Petz_Augustin} is called the \emph{Augustin mean} \cite{Nakiboglu2019, CHT19, CGH18}. 

The domain of the objective function is 
\begin{align*}
\dom f_{\mathrm{Petz}} = \left\{\sigma_B \in \mathcal{D}(\mathcal{H}_B) \mid \supp(\rho_B^x) \subseteq \supp(\sigma_B), \; \forall x, P_X(x) > 0
\right\}.
\end{align*}
Under the assumption that $\cup_{x:P_X(x)>0} \supp\left(\rho_B^x\right) = \mathcal{H}_B$, $\dom f_{\mathrm{Petz}}$ is then equal to the set of full-rank density matrices on $\mathcal{H}_B$.

Following the derivations given in \cite[(122)]{CHT19} or using \cite[Lemma III.10]{MO21} with the power function $h_{\alpha,z}(u):=u^{\frac{1-\alpha}{z}}$ and $z=1$, the gradient of \( f_{\text{Petz}} \) at $\sigma_B \in \dom f_{\mathrm{Petz}}$ is given by
\begin{align}
\begin{split} \label{eq:Petz_gradient}
\nabla\, f_{\mathrm{Petz}}(\sigma_B) &=    \frac{1}{\alpha-1} \mathbb{E}_{P_X} \frac{ \mathsf{D} h_{\alpha,1}(\sigma_B) \left[ (\rho_B^X)^\alpha \right] }{ \Tr\left[ (\rho_B^X)^\alpha \sigma_B^{1-\alpha} \right] } \\
&= \frac{1}{\alpha-1} \mathbb{E}_{P_X} \frac{ \sum_{a,b} h_{\alpha,1}^{[1]}(a,b) P_a^{\sigma_B} (\rho_B^X)^\alpha  P_b^\sigma  }{ \Tr\left[ (\rho_B^X)^\alpha \sigma_B^{1-\alpha} \right] },
\end{split}
\end{align}
where $h_{\alpha,z}^{[1]}:[0,1]^2\to \mathbb{R}$ is the first-order divided difference function of the power function $h_{\alpha,z}$, i.e.~
\begin{align}
h_{\alpha,z}^{[1]}(a,b) := \begin{cases}
\frac{h_{\alpha,z}(a)-h_{\alpha,z}(b)}{a-b}, & a\neq b \\
h_{\alpha,z}'(a), & a=b
\end{cases},
\end{align}
and $\{P_a^{\sigma_B}\}_a$ denote the eigen-projections of $\sigma_B$. 

\smallskip
The order-$\alpha$ \emph{sandwiched Augustin Information} for $\rho_{XB}$ is defined by 
\begin{equation} \label{eq:sand_Augustin}
     \min_{\sigma_B \in \mathcal{D}(\mathcal{H}_B)} f_{\mathrm{sandA}}(\sigma_B),\quad f_{\text{sandA}} ( \sigma_B ) \coloneqq \mathbb{E}_{P_X}{D_{\alpha}^*\left(\rho_B^X \vert\vert\sigma_{B}\right)} , 
\end{equation}
where $D^*_{\alpha}(\rho||\sigma)$ for density matrices $\rho,\sigma$ is the \emph{sandwiched R\'enyi divergence} \cite{MDS+13, WWY14}:
\begin{equation} \label{eq:sand}
    D_{\alpha}^{*}\left(\rho||\sigma\right) \coloneqq \frac{1}{\alpha-1} \log\operatorname{Tr}\left[ \left( \sigma^\frac{1-\alpha}{2\alpha}\rho\sigma^\frac{1-\alpha}{2\alpha} \right)^{\alpha}\right],\quad \alpha \in [\sfrac{1}{2}, \infty) \backslash \{1\},
\end{equation}
when the support of $\rho$ is contained in the the support of $\sigma$, and it is defined to be $+\infty$ otherwise.

Following \cite[Lemma III.10]{MO21} with $z=\alpha$, the gradient of \( f_{\text{sandA}} \) at $\sigma_B \in \dom f_{\mathrm{sandA}}$ is given by
\begin{align}
& \nabla\, f_{\mathrm{sandA}}(\sigma_B) \\
    & \quad = \frac{\alpha}{\alpha-1}\mathbb{E}_{P_X} 
    \frac{ \sigma_B^{-\frac{1-\alpha}{2\alpha}} \mathsf{D} h_{\alpha,\alpha}(\sigma_B)\left[\left( \sigma_B^\frac{1-\alpha}{2\alpha}\rho_B^X \sigma_B ^\frac{1-\alpha}{2\alpha} \right)^{\alpha} \right] \sigma_B^{-\frac{1-\alpha}{2\alpha}} }{ \operatorname{Tr}\left[ \left( \sigma_B^\frac{1-\alpha}{2\alpha}\rho_B^X \sigma_B ^\frac{1-\alpha}{2\alpha} \right)^{\alpha} \right] } \\
    & \quad = \frac{\alpha}{\alpha-1}\mathbb{E}_{P_X} 
    \frac{ \sigma_B^{-\frac{1-\alpha}{2\alpha}} \sum_{a,b} h_{\alpha,\alpha}^{[1]}(a,b) P_a^{\sigma_B} \left( \sigma_B^\frac{1-\alpha}{2\alpha}\rho_B^X \sigma_B ^\frac{1-\alpha}{2\alpha} \right)^{\alpha} P_b^{\sigma_B} \sigma_B^{-\frac{1-\alpha}{2\alpha}} }{ \operatorname{Tr}\left[ \left( \sigma_B^\frac{1-\alpha}{2\alpha}\rho_B^X \sigma_B ^\frac{1-\alpha}{2\alpha} \right)^{\alpha} \right] }.
\end{align}

The following proposition shows that if we directly adopt entropic mirror descent to compute the Petz-Augustin and sandwiched Augustin information, then existing convergence guarantees for entropic mirror descent in \cite{Auslender2006,Bauschke2017,Beck2003,Lu2018,Teboulle2018} do not apply. 

\begin{proposition} \label{proposition:unbounded_Augustin}
    Both the gradients and Hessians of $f_{\mathrm{Petz}}$ and $f_{\mathrm{sandA}}$ are not bounded. The functions $f_{\mathrm{Petz}}$ and $f_{\mathrm{sandA}}$ are not smooth relative to the von Neumann entropy.
\end{proposition}

\begin{proof}
Consider the two-dimensional case, where $\sigma = (\sigma_{i,j})_{1\leq i,j\leq 2} \in \mathbb{C}^{2 \times 2}.$ Suppose that there are only two summands, with $A_1=$ $\begin{pmatrix}
  1 & 0\\ 
  0 & 0
\end{pmatrix}$
and $A_2=$ 
$\begin{pmatrix}
  0 & 0\\ 
  0 & 1
\end{pmatrix}.$ 
Then, we have $f(\sigma) = \frac{1}{\alpha-1}(\log \sigma_{1,1}^{1-\alpha} + \log \sigma_{2,2}^{1-\alpha})$.  
It suffices to disprove all properties for the specific $f$ on the set of diagonal density matrices. 
Hence, we focus on the function $g(x, y) = \frac{1}{\alpha-1}(\log x^{1-\alpha} + \log y^{1-\alpha}),$ defined for any $(x,y)$ in the probability simplex. 
The rest of the proof follows that of Proposition A.1 in \cite{Li2019a}.
\end{proof}


\subsection{Conditional Sandwiched R\'enyi Entropy and Sandwiched R\'enyi Information}

Let $\mathcal{H}_A$ be a finite-dimensional Hilbert space, and let $\rho_{AB}$ be a bipartite density matrix on $\mathcal{H}_{A} \otimes \mathcal{H}_{B}$
Without loss of generality, we further assume that the reduced density matrix $\rho_B$ on $\mathcal{H}_B$ has full support.
For any $\tau_A$ being a density matrix on $\mathcal{H}_A$ with $\supp(\rho_A)\subseteq \supp(\tau_A)$, we define the order-$\alpha$ \emph{generalized sandwiched R\'enyi information} for $\rho_{AB}$ and $\tau_A$ as
\begin{equation}
    \label{Unifying equations}
    I_{\alpha}^{*}(\rho_{AB}|\tau_{A})_{\rho}\coloneqq \mathop{\min}_{\sigma_B \in \mathcal{D}(\mathcal{H}_{B})} D_{\alpha}^{*}(\rho_{AB}|| \tau_{A} \otimes \sigma_{B}),
\end{equation}
where $D_{\alpha}^{*}(\rho||\sigma)$ is defined in \eqref{eq:sand}.
Adapting the derivations given in \cite[Lemma 22]{HT14}, we obtain the gradient at $\sigma_B \in \dom D_{\alpha}^{*}\left(\rho_{AB}|| \tau_{A} \otimes \sigma_{B} \right)$:
\begin{align}
    &\nabla D_{\alpha}^{*}\left(\rho_{AB}|| \tau_{A} \otimes \sigma_{B} \right) \\
    &= \frac{\alpha}{\alpha-1}\frac{ \sigma_B ^{-\frac{1-\alpha}{2\alpha}}
    \mathsf{D} h_{\alpha,\alpha}(\sigma_B)  \left[\operatorname{Tr}_A\left[\left((\tau_{A} \otimes \sigma_{B})^{\frac{1-\alpha}{2\alpha}}\rho_{AB}(\tau_{A} \otimes \sigma_{B} )^{\frac{1-\alpha}{2\alpha}}\right)^{\alpha}\right] \right]\sigma_B^{-\frac{1-\alpha}{2\alpha}}}{\operatorname{Tr}\left[\left((\tau_{A} \otimes \sigma_{B})^{\frac{1-\alpha}{2\alpha}}\rho_{AB}(\tau_{A} \otimes \sigma_{B})^{\frac{1-\alpha}{2\alpha}}\right)^{\alpha}\right]} \\
    &= \frac{\alpha}{\alpha-1}\frac{ \sigma_B ^{-\frac{1-\alpha}{2\alpha}}
    \sum_{a,b} h_{\alpha,\alpha}^{[1]}(a,b) P_a^{\sigma_B} \operatorname{Tr}_A\left[\left((\tau_{A} \otimes \sigma_{B})^{\frac{1-\alpha}{2\alpha}}\rho_{AB}(\tau_{A} \otimes \sigma_{B} )^{\frac{1-\alpha}{2\alpha}}\right)^{\alpha}\right] P_b^{\sigma_B} \sigma_B^{-\frac{1-\alpha}{2\alpha}}}{\operatorname{Tr}\left[\left((\tau_{A} \otimes \sigma_{B})^{\frac{1-\alpha}{2\alpha}}\rho_{AB}(\tau_{A} \otimes \sigma_{B})^{\frac{1-\alpha}{2\alpha}}\right)^{\alpha}\right]}.
\end{align}

If $\tau_{A}$ is the identity matrix $I_A$ on $\mathcal{H}_A$, the negation of the generalized sandwiched R\'enyi information is the
order-$\alpha$ \textit{conditional sandwiched R\'enyi entropy} (its negation is called the order-$\alpha$ R\'enyi coherent information):
\begin{equation}
    \label{Conditional sandwiched Renyi Entropy}
H_{\alpha}^{*}(A|B)_{\rho}\coloneqq -\mathop{\min}_{\sigma_B \in \mathcal{D}(\mathcal{H}_{B})} f_{\mathrm{cond}}(\sigma_B) , \quad f_{\mathrm{cond}}(\sigma_B) \coloneqq D_{\alpha}^{*}(\rho_{AB}|| I_A \otimes \sigma_{B}) . 
\end{equation}
The conditional sandwiched Re\'nyi entropy can be negative and is lower bounded by $-\log \min\{|\mathcal{H}_A|, |\mathcal{H}_B|\}$ \cite[Lemma 5.2]{Tom16}.

If $\tau_{A}=\rho_A \coloneqq \operatorname{Tr}_{B}[\rho_{AB}]$, where $\operatorname{Tr}_{B}$ denotes the partial trace with respect to $\mathcal{H}_{B}$, then the generalized sandwiched R\'enyi information reduces to the order-$\alpha$ \textit{sandwiched R\'enyi Information}:
\begin{equation}
    \label{Sandwiched Renyi Information}
I_{\alpha}^{*}(A:B)_{\rho} \coloneqq \mathop{\min}_{\sigma_B \in \mathcal{D}(\mathcal{H}_{B})} f_{\mathrm{sandR}}(\sigma_B) , \quad f_{\mathrm{sandR}}(\sigma_B) \coloneqq D_{\alpha}^{*}(\rho_{AB}|| \rho_{A} \otimes \sigma_{B}) . 
\end{equation}
When the system $A = X$ is classical, the density matrix $\rho_{XB} = \oplus_{x} P_{X}(x)\rho_{B}^{x}$ possesses a direct sum structure. Then, \eqref{Sandwiched Renyi Information} can be written as:
\begin{equation*}
    I_{\alpha}^{*}(X:B)_{\rho} =  \mathop{\min}_{\sigma_B \in \mathcal{D}(\mathcal{H}_{B})}\frac{1}{\alpha-1} \log\mathbb{E}_{P_X}
    \operatorname{Tr}\left[ \left( \sigma^\frac{1-\alpha}{2\alpha}\rho_{B}^{X}\sigma^\frac{1-\alpha}{2\alpha} \right)^{\alpha}\right].
\end{equation*}
Note that this is very similar to the sandwiched Augustin information given in \eqref{eq:sand_Augustin} except that the expectation $\mathbb{E}_{P_X}$ is taken inside $\log$ instead of outside $\log$.


The following proposition shows that if we directly adopt entropic mirror descent to compute the conditional sandwiched R\'enyi entropy and sandwiched R\'enyi information, then existing convergence guarantees for entropic mirror descent in \cite{Auslender2006,Bauschke2017,Beck2003,Lu2018,Teboulle2018} do not apply.

\begin{proposition} \label{proposition:unbounded_conditional}
    Both the gradients and Hessians of the functions $f_{\mathrm{cond}}$ and $f_{\mathrm{sandR}}$ are not bounded. The functions $f_{\mathrm{cond}}$ and $f_{\mathrm{sandR}}$ are not smooth relative to the von Neumann entropy.
\end{proposition}
The proof is similar to that of Proposition \ref{proposition:unbounded_Augustin} and hence omitted.

\section{Optimization Preliminaries}
\label{sec:format}

We introduce basic notions in optimization theory that are necessary to understanding this paper. 
In particular, we introduce mirror descent and the standard Polyak step size. 

%
%

\subsection{Bregman Divergence}
\label{subsec:Bregman}

Let \( E \) be a finite-dimensional space endowed with a norm \( \norm{ \,\cdot\, } \). 
Let \( h: E \to \interval{- \infty}{\infty} \) be a convex function differentiable on \( \dom \nabla h \), assumed to be non-empty. 
The \emph{Bregman divergence} associated with \( h \) is defined by
\[
D_h ( x, y ) \coloneqq h ( x ) - \left[ h ( y ) + \braket{ \nabla h ( y ), x - y } \right] . 
\]
That is, \( D_h ( x, y ) \) equals the difference between \( h ( x ) \) and its first-order approximation at \( y \). 
By the convexity of \( h \), the Bregman divergence \( D_h \) is always non-negative. 
We say that \( h \) is \emph{strongly convex} with respect to the norm \( \norm{ \cdot } \) if
\[
D_h ( x, y ) \geq \frac{1}{2} \norm{ x - y } ^ 2 . 
\]
Obviously, if \( h \) is strongly convex, the \emph{level set} \( \set{ y \in \dom \nabla h | D_{h} ( x, y ) \leq \gamma } \) is bounded for any given \( x \in \dom h \) and \( \gamma \geq 0 \). 

The following two choices are perhaps the most popular. 
\begin{enumerate}
\item If we choose \( h \) to be the square Euclidean norm, obviously strongly convex with respect to the Euclidean norm, then the associated Bregman divergence is the square Euclidean distance. 
\item If we choose \( h \) to be the negative von Neumann entropy, strongly convex with respect to the Schatten \(1\)-norm by Pinsker's inequality \cite{Hiai1981}, then the associated Bregman divergence is the quantum relative entropy. 
\end{enumerate}

%
%
%

\subsection{Mirror Descent}
\label{sec:mirror_descent}

Consider the problem of minimizing a convex differentiable function \( f \) on a closed convex set \( \mathcal{Z} \subseteq E \). 
Let \( h \) be a convex function differentiable on \( \inte ( \dom h ) \) such that the closure of \( \dom \nabla h \) contains \( \mathcal{Z} \). 
Mirror descent iterates as follows. 
\begin{itemize}
\item Let the initial iterate \( z_1 \in \mathcal{Z} \). 
\item For each iteration index \( t \in \mathbb{N} \), 
\[
z_{t + 1} \in \mathop{\arg\!\min}_{z \in \mathcal{Z}} \Set {\eta_t \braket{ \nabla f ( z_t ), z - z_t } + D_h ( z, z_t )}, 
\]
for some properly chosen step size \( \eta_t > 0 \). 
\end{itemize}

The following two setups are perhaps the most popular. 
\begin{enumerate}
\item When \( E \) is an Euclidean space and \( h \) is the square Euclidean norm, mirror descent is equivalent to projected gradient descent. 
\item When the constraint set \( \mathcal{Z} \) is the set of quantum density matrices, a standard setup of mirror descent is to choose \( h \) as the von Neumann entropy. 
Then, the iteration rule has a closed form (see, e.g., \cite{Bubeck2015a,BenTal2015}): 
\begin{equation}
z_{t + 1} = c^{-1} \exp \left( \log ( z_t ) - \eta_t \nabla f ( z_t ) \right) , 
\end{equation}
where \( \exp \) and \( \log \) are matrix exponential and logarithm, respectively, and $c$ is a positive number chosen such that \( \mathop{\text{trace}} ( z_{t + 1} ) = 1 \). 
The resulting method is called entropic mirror descent. 
\end{enumerate}

Notice that the iterates of entropic mirror descent are by definition full-rank; 
then, the stalling issue discussed in Section \ref{sec:intro} vanishes. 
Moreover, if the function \( f \) has a bounded gradient, then the convergence speed of mirror descent is only logarithmically dependent on the parameter dimension (see, e.g., \cite{BenTal2015}), a desirable property for quantum information applications. 
We will adopt entropic mirror descent with \( \eta_t \) chosen by a variant of the Polyak step size for minimizing the quantum R\'enyi divergences, though our convergence guarantee applies to general mirror descent setups. 

\subsection{Polyak Step Size}

Consider minimizing a convex differentiable function \( f \) on an Euclidean space by gradient descent: 
\[
z_{t + 1} = z_t - \eta_t \nabla f ( z_t ) , 
\]
for some step size \( \eta_t > 0 \). 
Suppose the optimal value \( f^\star \coloneqq \min_z f ( z ) \) is known. 
The Polyak step size suggests choosing \( \eta_t \) following
\begin{equation} \label{eq:eta}
    \eta_{t} = \frac{f(\inp_{t})- f^\star}{ \norm{ \nabla f( z_t ) }^2 },
\end{equation}
where the norm is the Euclidean norm. 
Let \( z^\star \) be a minimizer. 
To motivate the step size, notice that by the convexity of \( f \), we have
\begin{align*}
    ||\inp_{t+1}-\inp^{\star}||^{2}
    & = ||\inp_{t}-\eta_{t}\nabla f(\inp_{t})-\inp^{\star}||^{2} \\
    & \leq ||\inp_{t}-\inp^{\star}||^{2}-2\eta_{t}(f(\inp_t)-f^{\star})+\eta_{t}^{2}||\nabla f(\inp_{t})||^{2} . 
\end{align*}
The Polyak step size minimizes the upper bound as a function of \( \eta_t \). 

Notice that the Polyak step size involves the optimal value \( f^\star \), which is not accessible in general. 
In the next section, we modify an ``adjustment'' of the Polyak step size proposed in \cite{Nedic2001}, and extend it for possibly non-Euclidean norms. 

\section{Mirror Descent with Polyak Step Size and Its Convergence}
\label{sec:pagestyle}

We describe the proposed method, mirror descent with the Polyak step size. 
Then, we present a convergence guarantee of the method. 

\subsection{Proposed Method}

Consider the general convex optimization problem: 
\begin{equation}
\label{eq:1}
\min_{z \in \mathcal{Z}} f ( z ) , 
\end{equation}
where \( \mathcal{Z} \) is a closed convex set in a finite-dimensional space \( E \), and the function \( f \) is convex on \( \mathcal{Z} \) and differentiable on \( \inte ( \dom f ) \). 
The space \( E \) is endowed with a norm \( \norm{ \cdot } \) and the dual space is endowed with the dual norm \( \norm{ \cdot }_* \). 
We assume a minimizer \( z^\star \in \mathcal{Z} \) exists. 
Notice the four optimization problems introduced in the previous section all fit in this problem template. 

The proposed method, mirror descent with the Polyak step size, iterates as mirror descent (see Section \ref{sec:mirror_descent}) with the following specific step size selection rule. 
\begin{itemize}
\item Let \(  \delta_1 \geq \delta > 0 \), \( \beta < 1 \), \( \gamma \geq 1 \), and \( c > 0 \). 
The parameter \( \delta \) specifies the desired optimization error (see Theorem \ref{Convergence_Alg:1}). 
\item For each \( t \in \mathbb{N} \), first, compute an estimate \( \tilde{f}_t \) of \( f ( z^\star ) \) following 
\begin{equation}
    \tilde f_t \coloneqq \min_{1 \leq t^\prime \leq t} f(\inp_{t^\prime})-\delta_{t} . 
    \label{eq:4}
\end{equation}
Then, compute the modified Polyak step size 
\begin{equation}
    \eta_{t} = \frac{f(\inp_{t})- \tilde f_t}{c \| \nabla f(\inp_{t})\|_{*}^{2}},
    \label{eq:3}
\end{equation}
for some constant \( c > 0 \). 

Finally, after obtaining the next iterate \( z_{t + 1} \), compute \( \delta_{t + 1} \) following
\begin{equation}
\label{eq:delta}
\delta_{t + 1} = 
\begin{cases}
\gamma \delta_t & \text{if } f ( z_{t + 1} ) \leq \tilde{f}_t , \\
\max \set{ \beta \delta_t, \delta } & \text{otherwise} . 
\end{cases}
\end{equation}
\end{itemize}

This step size selection rule mimics the ``first adjustment'' of the Polyak step size in \cite{Nedic2001}, which replaces the optimal value \( f ( z^\star ) \) with a sequence of its estimates. 
The only difference is that we allow a general norm here, whereas the norm is Euclidean in \cite{Nedic2001}. 
This slight difference is necessary for mirror descent-type methods. 
The norm should be chosen in accordance with the chosen Bregman divergence in our convergence analysis, as in standard analyses of mirror descent \cite{Auslender2006,Beck2003}.

\subsection{Convergence Analysis}

We start with two existing results that will be used in the analysis.
The first is standard in the analyses of mirror descent-type methods. 
See, e.g, \cite{Juditsky2012,Teboulle2018} for its proof. 

\begin{lemma}
\label{Dh_decrease}
    Consider solving the general convex optimization problem \eqref{eq:1} by mirror descent as specified in Section \ref{sec:mirror_descent}. 
    Suppose that $h$ is strongly convex with respect to $||\cdot||$ on $\inpS$. 
    Then, for any $t$ and $\inp \in \inpS \bigcap \dom h$, 
    \begin{equation*}
        D_{h}(\inp,\inp_{t+1})\leq D_{h}(\inp,\inp_{t})-\eta_{t}  \left(f(\inp_{t})-f(\inp) \right)
        +\frac{\eta_{t}^{2}}{2}\left\|\nabla f(\inp_{t})\right\|_{*}^{2}.
    \end{equation*}
\end{lemma}

The second provides a useful property of Bregman divergences defined by convex and essentially smooth functions.

\begin{lemma}[Thm. 3.8 (i) in \cite{Bauschke1997}]
$\label{Dh_diverge}$
Let $h$ be a closed convex funtion, proper on $\dom h$, differentiable on $\inte (\dom h)$, and essentially smooth. Suppose that 
\begin{enumerate}[label = (\roman*)]
    \item $z \in \inte (\dom h)$ and the sequence $\set {z_t}_{t = 1}^{\infty} \subseteq \inte (\dom h),$
    \item  $\set {z_t}_{t = 1}^{\infty}$ converges to $z$ lying on the boundary of $\dom h$. 
\end{enumerate}
Then $\set {D_h(z,z_t)}_{t=1}^{\infty} \to \infty$.
\end{lemma}

\begin{remark}
We will use the lemma in the following manner: 
If a sequence \( \set{ z_t }_{t = 1}^\infty \) satisfies \( D_h ( z, z_t ) < + \infty \) for some \( z \in \inte ( \dom h ) \) and all \( t \), then the sequence does not approach the boundary of \( \dom h \).
\end{remark}

The following theorem guarantees the convergence of mirror descent with the Polyak step size up to a prescribed numerical accuracy. 

\begin{theorem}
    \label{Convergence_Alg:1}
	Consider solving \eqref{eq:1} by mirror descent with the Polyak step size. 
	Suppose that 
	\begin{itemize}
	\item \( \min_{z \in \mathcal{Z}} f ( z ) > - \infty \), 
	\item The set of minimizers \( \mathcal{Z}^{\star}\) lies in the closure of \(\dom \nabla h \), 
	\item the sequence of iterates \( \set{ z_t }_{t = 1}^\infty \) is well-defined,
	\item the function \( h \) is essentially smooth and strongly convex with respect to the norm \( \norm{ \cdot } \) on \( \mathcal{Z} \), and 
	\item \( \nabla f \) is bounded on any compact subset of \( \mathcal{Z} \cap \dom \nabla h \).
	\end{itemize}
	If we choose \( c > 1 / 2 \), then 
	\[
	\inf_{t \geq 1} f ( z_t ) \leq f ( z^\star ) + \delta . 
	\]
    
\end{theorem}

\begin{remark}
The first three assumptions ensure that the problem and proposed method are well defined. 
See, e.g., \cite{Bauschke2017} for a discussion.
The fourth assumption implies that the function \( h \) is Legendre. 
This assumption is standard in setting up mirror descent-type methods.
The last assumption assumes a locally bounded gradient.
This assumption enables us to establish the convergence guarantee without the bounded gradient/Hessian and relative smoothness conditions.
\end{remark}

\begin{proof}

We prove by contradiction. Suppose that
\begin{equation}
    \inf_{t \geq 1}f(\inp_t) \geq f(\inp^{\star}) + \delta + \varepsilon, \label{eq:5}
\end{equation}
for some \( \varepsilon > 0 \).

We can observe that the inequality \( f(\inp_{t+1}) > \tilde f_t \) must hold for infinitely many iterations. 
Otherwise, by \eqref{eq:delta}, we have \( \inf_{t \geq 1} f(\inp_t) = - \infty \), contradicting the assumption that $f(\inp^{\star})>-\infty$.
Then, by the iteration rule, there exists some \( \Bar{t} \in \mathbb{N} \) such that 
\begin{equation}
    \delta_t = \delta, \quad \forall t \geq \Bar{t}. \label{eq:6}
\end{equation}
By $\eqref{eq:5}$ and $\eqref{eq:6}$, we have
\begin{align}
        \tilde f_t &= \min_{1 \leq t^\prime \leq t} f(\inp_{t^\prime})-\delta \nonumber\\
        &\geq \inf_{t^\prime \geq 1}f(\inp_{t^\prime})-\delta \geq f(\inp^{\star}) + \varepsilon, \quad \forall t \geq \bar{t}.
\end{align}
Since the objective function \( f \) is continuous, we can find some \( \bar{z} \in \mathcal{Z} \) such that
    \begin{equation}
    \label{f_tilde_t}
        \tilde f_t \geq f(\bar{z}),\quad \forall t \geq \bar t.
    \end{equation}
By Lemma~\ref{Dh_decrease} with $\inp = \Bar{\inp}$ and the definition of the step size, we obtain for any $t\geq\Bar{t}$,
\begin{align}
        D_{h}(\Bar{\inp},\inp_{t+1})&\leq D_{h}(\Bar{\inp},\inp_{t})-\eta_{t} \left(f(\inp_{t})-f(\Bar{\inp})\right) \nonumber +\frac{\eta_{t}^{2}}{2}\left\|\nabla f(\inp_{t})\right\|_{*}^{2}\nonumber\\
        &\leq D_{h}(\Bar{\inp},\inp_{t})-\eta_{t} (f(\inp_{t})-\tilde{f_t})\nonumber +\frac{\eta_{t}^{2}}{2}\left\|\nabla f(\inp_{t})\right\|_{*}^{2}\nonumber\\
        &=D_{h}(\Bar{\inp},\inp_{t})- \left(1-\frac{1}{2c}\right)\left(\frac{f(\inp_t)-\tilde{f_t}}{\left\|\nabla f(\inp_{t})\right\|_{*}}\right)^{2}, \label{eq:8}
\end{align}
where the second inequality follows from $\eqref{f_tilde_t}$ and the equality follows from the definition of the step size. 
A telescopic sum gives
\begin{equation}
    \left(1-\frac{1}{2c}\right)\sum_{t = \Bar{t}}^{\infty}\left(\frac{f(\inp_t)-\tilde{f_t}}{\left\|\nabla f(\inp_{t})\right\|_{*}}\right)^{2}\leq D_h(\Bar{\inp},\inp_{\Bar{t}})< \infty. \label{eq:9}
\end{equation}

It remains to show that \( \set{ \norm{ \nabla f ( z_t ) }_* }_{t = 1}^\infty \) is bounded from above by some \( M > 0 \).
If this is actually the case, then
\begin{equation*}
    \sum_{t = \Bar{t}}^{\infty}\left(\frac{f(\inp_t)-\tilde{f_t}}{\left\|\nabla f(\inp_{t})\right\|_{*}}\right)^{2} \geq
    \sum_{t = \Bar{t}}^{\infty}\frac{\delta^{2}}{M^{2}} = \infty,
\end{equation*}
contradicting $\eqref{eq:9}$; 
then, it must hold that $\inf_{t \geq 1}f(\inp_t) \leq f(\inp^{\star})+\delta$. 

By \eqref{eq:8}, we know that the sequence \( \set{ D_h ( \Bar{\inp}, z_t ) }_{t = 1}^\infty \) is bounded. 
Since the function \( h \) is strongly convex, the sequence of iterates \( \set{ z_t }_{t = 1}^\infty \) is bounded in \( \mathcal{Z} \cap \dom \nabla h \). 
Unfortunately, the set \( \mathcal{Z} \cap \dom \nabla h \) may not be closed as \( \dom \nabla h \) may not be closed (e.g., in entropic mirror descent), so this does not imply the boundedness of \( \set{ \norm{ \nabla f ( z_t ) }_* }_{t = 1}^\infty \) even if \( f \) is continuously differentiable. 
Nevertheless, Lemma \ref{Dh_diverge} ensures that the sequence \( \set {z_t}_{t = 1}^\infty \) does not approach the boundary of \( \dom \nabla h \).
Therefore, we can find a compact (bounded and closed) set in \(\mathcal{Z} \cap \dom \nabla h \) that contains \( \set {z_t}_{t = 1}^\infty \).
Then, by the last assumption of the theorem, the sequence \( \set{ \norm{ \nabla f ( z_t ) }_* }_{t = 1}^\infty \) must be bounded from above.


\end{proof}

\begin{remark}
The convergence analysis indeed applies to convex non-differentiable objective functions, as long as there are locally bounded subgradients.
\end{remark}

\subsection{Applicability of Theorem \ref{Convergence_Alg:1} to Minimizing \\ Quantum R\'enyi Divergences }
\label{Applicability_of_thm}
The following proposition helps us to explain why the assumptions in Theorem \ref{Convergence_Alg:1} hold for our considered problems in Section~\ref{sec:problem}.
\begin{proposition}
\label{proposition: grad_bounded_by_eigenvalue}
    The gradients of $f_{\mathrm{Petz}},$ $f_{\mathrm{sandA}},$ $f_{\mathrm{cond}}$ and $f_{\mathrm{sandR}}$ are upper bounded by the reciprocals of their minimum eigenvalues of $\sigma_B$ in the Schatten $1$-norm.
\end{proposition}

\begin{proof}
Using the cyclic property of trace and H\"older's inquality, one has
\begin{align}
    \left\| \nabla\, f_{\mathrm{Petz}}(\sigma_B) \right\|_1
    &= \left\langle \mathbb{E}_{P_X} \frac{ \sigma_B^{\frac{1-\alpha}{2}} (\rho_B^X)^\alpha \sigma_B^{\frac{1-\alpha}{2}} }{ \Tr\left[(\rho_B^X)^\alpha \sigma_B^{1-\alpha}\right]}, \sigma_B^{-1} \right\rangle \\
    &\leq \left\|\sigma_B^{-1} \right\|_\infty.
\end{align}
Following similar reasoning, we have 
\begin{align}\left\| \nabla\, f_{\mathrm{sandA}}(\sigma_B) \right\|_1\leq \left\|\sigma_B^{-1} \right\|_\infty,
\end{align}
and
\begin{align}
\left\|\nabla D_{\alpha}^{*}\left(\rho_{AB}|| \tau_{A} \otimes \sigma_{B} \right)\right\|_1 \leq \left\|\sigma_B^{-1}\right\|_\infty.
\end{align}
\end{proof}

One can observe that Theorem \ref{Convergence_Alg:1} applies to solving the problems in Section \ref{sec:problem} by entropic mirror descent with the Polyak step size. 
The first assumption holds because the functions $f_{\mathrm{Petz}}, f_{\mathrm{sandA}}, f_{\mathrm{sandR}}$  are lower bounded by 0, and the function $-f_{\mathrm{cond}}$ is lower bounded by $-\log \min\{|\mathcal{H}_A|, |\mathcal{H}_B|\}$ (see Section~\ref{sec:problem}).
The second assumption holds as the closure of \( \dom \nabla h \) corresponds to the set of positive semidefinite matrices, which contains the set of density matrices \( \mathcal{D} \).
The third assumption holds because entropic mirror descent always generates full-rank iterates.
The fourth assumption is satisfied: the gradient of the von Neumann entropy is locally bounded on its domain, which implies the essential smoothness property; the strong convexity property is already discussed in Section \ref{subsec:Bregman}.
We verify the last assumption as follows. Since $\mathcal{Z} \cap \dom \nabla h$ consists of full-rank matrices, the minimum eigenvalue of any density matrix in a compact subset of $\mathcal{Z} \cap \dom \nabla h$ is bounded away from zero. Hence, the gradient of $f$ is bounded according to Propositions~\ref{proposition: grad_bounded_by_eigenvalue}.


\section{Numerical Results} \label{sec:numerics}

We report the numerical results of computing the four information-theoretic quantities in Section \ref{sec:problem} by entropic mirror descent with the Polyak step size. 
We compare the convergence speeds with those of entropic mirror descent with Armijo line search---recall it is also guaranteed to converge (see Section \ref{sec:related_work}). 
We tried tuning the parameters in both methods to get fast convergence speeds. 
\begin{figure}[ht]
    \centering
     \includegraphics[width = 12cm]{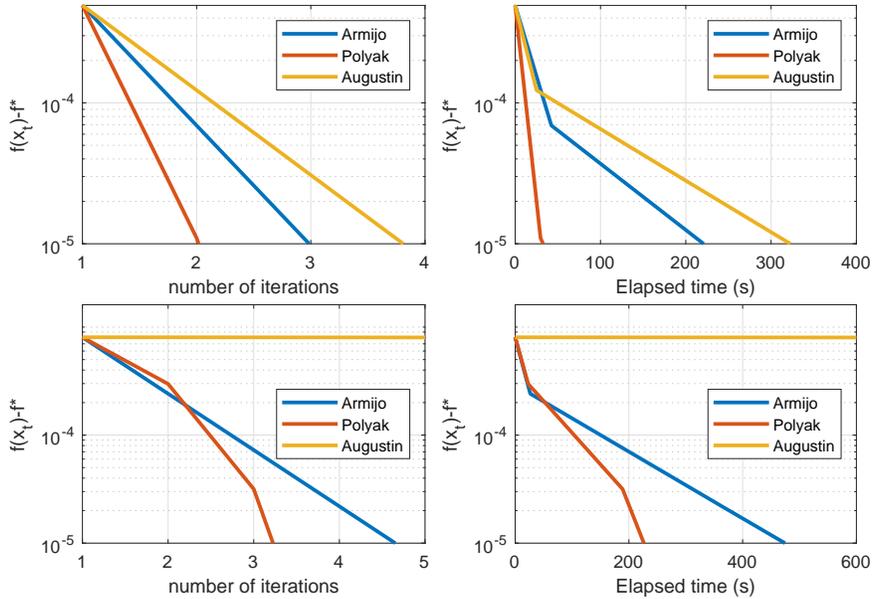}
     \caption{Convergence speeds for computing the Petz--Augustin information of order $\alpha = 0.5$ (upper) and $\alpha = 2$ (bottom).}
     \label{Petz_Aug_Convergence_speed}
\end{figure}
%

Figure \ref{Petz_Aug_Convergence_speed} shows the results of computing the Petz--Augustin information for \( \alpha = 0.5 \) and \( \alpha = 2 \). 
The support cardinality of \( X \) is \( 2^{10} \) and the dimension of \( \mathcal{H}_B \) equals \( 2^{7} \) for both cases. 
The ensemble of quantum states, $\left\{ \rho_B^x \right\}_x$, is generated by the \href{http://www.qetlab.com/RandomDensityMatrix}{RandomDensityMatrix} function of the MATLAB toolbox QETLAB. The Fr\'echet derivative of the power function is computed by the MATLAB package given in \cite{Higham2013}.
We set $\Bar{\alpha}=10$, $r = \tau = 1 / 2$ for $\alpha = 0.5$, and $\Bar{\alpha}=8$, $r = \tau = 0.7$ for $\alpha = 2$ in entropic mirror descent with Armijo line search, following the notations in \cite{Li2019a}. 
We set $\delta_1 = 2.5$ for $\alpha = 0.5$, and $\delta_1 = 1$ for $\alpha = 2$, and $\delta = 10^{-5}$, $\gamma = 1.25$, $\beta = 0.75$ for both cases in entropic mirror descent with the Polyak step size. 
We also show the numerical results of a direct extension of the fixed-point iteration proposed in \cite{Nakiboglu2019} to the quantum case, i.e.~
\begin{equation}
\sigma_B^{t+1} = - \left(\sigma_B^{t}\right)^{\frac12} \nabla\, f_{\mathrm{Petz}}\left(\sigma_B^{t}\right) \left(\sigma_B^{t}\right)^{\frac12}, \label{Augustin}
\end{equation}
for which we label as ``Augustin'' in the figure, though it lacks a convergence guarantee in the quantum case. 



\begin{figure}[ht]
    \centering
     \includegraphics[width = 12cm]{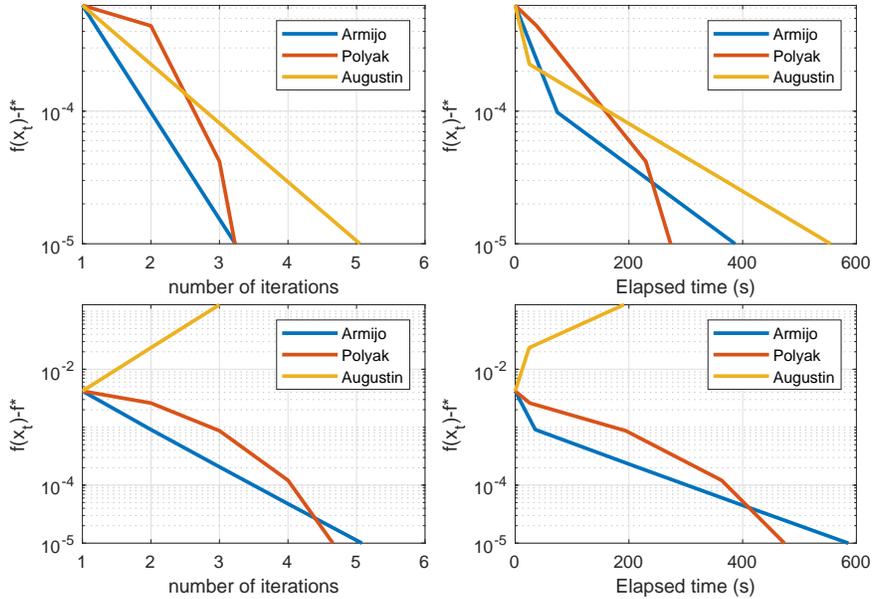}
     \caption{Convergence speeds for computing the sandwiched Augustin information of order $\alpha = 0.5$ (upper) and $\alpha = 10$ (bottom).}
     \label{Sand_Aug_Convergence_speed}
\end{figure}

Figure \ref{Sand_Aug_Convergence_speed} shows the results of computing the sandwiched Augustin information for \( \alpha = 0.5 \) and \( \alpha = 10 \). 
The support cardinality of $X$ is $2^{10}$ and the dimension of $\mathcal{H}_B$ equals $2^{7}$ for both cases. 
We set $\Bar{\alpha}=7$, $r = \tau = 0.6$ for \( \alpha = 0.5 \), and $\Bar{\alpha}=4, r = \tau = 0.7$ for $\alpha = 10$ in entropic mirror descent with Armijo line search. 
We set $\delta_1 = 5$ for $\alpha = 0.5$, and $\delta_1 = 1$ for $\alpha = 10$, and $\gamma = 1.3$, $\beta = 0.7$,  $\delta = 10^{-5}$ for both cases in entropic mirror descent with the Polyak step size. We also show the numerical results of a direct extension of the fixed-point iteration proposed in \cite{Nakiboglu2019} to the quantum case, for which we label as ``Augustin'' in the figure. The formula can be obtained by replacing $\nabla f_{\mathrm{Petz}}$ with $\nabla f_{\mathrm{sandA}}$ in \eqref{Augustin}. 

\begin{figure}[ht]
    \centering
     \includegraphics[width = 12cm]{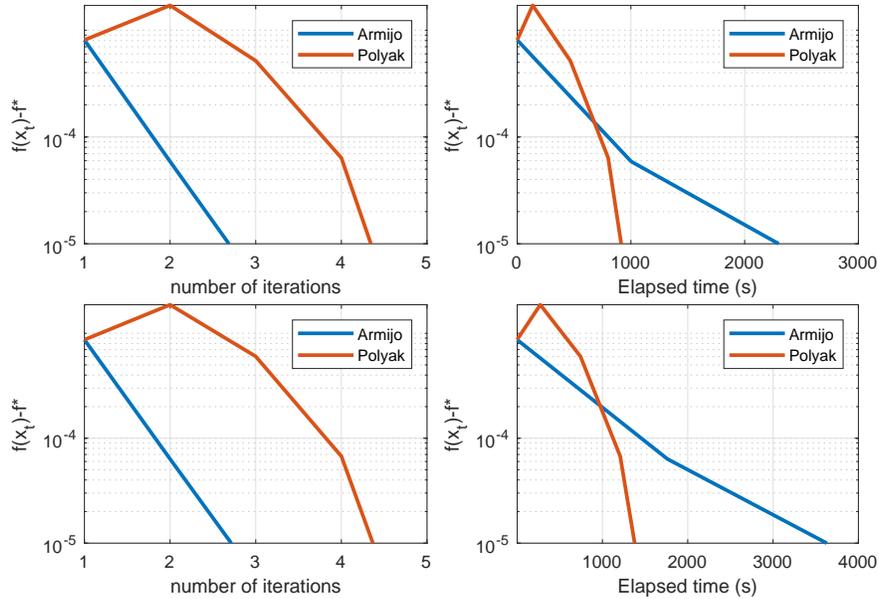}
     \caption{Convergence speeds for computing the conditional sandwiched R\'enyi entropy (upper) and the sandwiched R\'enyi information (bottom), both of order $\alpha = 10$.}
     \label{Cond. Sand. and Sand. Ren_Convergence_speed}
\end{figure}

Figure \ref{Cond. Sand. and Sand. Ren_Convergence_speed} shows the results of computing the conditional sandwiched R\'enyi entropy and sandwiched R\'enyi information, both of order \( \alpha = 10 \).
The dimensions of $\mathcal{H}_A$ and $\mathcal{H}_B$ are both $2^{6}$. 
We set $\Bar{\alpha}=8, r = \tau = 0.6$ in entropic mirror descent with Armijo line search. 
We set $\delta_1 = 1$, $\delta = 10^{-5}$, $\gamma = 1.3$, and $\beta = 0.7$ in entropic mirror descent with the Polyak step size. 



Since the optimal value is unknown, in the three figures, we approximate the optimal value by the result of entropic mirror descent with Armijo line search after 100 iterations. 
The approximations are denoted by \( f^* \) in the figures. 

Obviously, entropic mirror descent with the Polyak step size is competitive in terms of the elapsed time, though in terms of the number of iterations, it seems to be slower for computing the conditional sandwiched R\'enyi entropy and the sandwiched R\'enyi information. 
The reason is that entropic mirror descent with Armijo line search requires evaluating the function value at an indefinite number of points in each iteration, which introduces additional computational burden. 
The quantum extension of the fixed-point iteration in \cite{Nakiboglu2019} looks also competitive for computing the Petz-Augustin information and the sandwiched Augustin information of order \( 0.5 \), but fails to converge for the order-\( 2 \) case and the order-\( 10 \) case---recall that this heuristic extension lacks a convergence guarantee currently. 

\section{Concluding Remarks}
We have proposed mirror descent with the Polyak step size and established its convergence. 
Numerical results show that for minimizing R\'enyi divergences, entropic mirror descent with the Polyak step size is competitive in the convergence speed. 
A drawback of our method is that it only guarantees convergence to approximate solution, and the desired numerical accuracy needs to be given as a parameter of the method. 
We are working on improving the proposed method such that it provably converges to the optimum while keeping the current computational efficiency. 

Our convergence guarantee is applicable to many other optimization problems in quantum information studies, such as computing the relative entropy of any quantum resource \cite{VedralPlenioRippinKnight97, BG15, CG19}, optimizing quantum entropies with energy constraints \cite{Hol12}, minimizing the second argument of the $\varepsilon$-hypothesis testing divergence \cite[Definition~17]{MW14}, computing the quantum rate distortion function \cite{DHW13, DHWW13}, and computing the quantum information bottleneck function and related quantities \cite{DHW19, CDR19a, CDR21}.
To efficiently implement mirror descent with the Polyak step size in these problems, however, is non-trivial and left as a future work. 
%
%

In the numerical experiments, we always choose the initial iterates to be \( I / d \), where \(d\) is the dimension of \(\mathcal{H}_B\). 
Whereas \( I / d \) seems to be a natural choice and our convergence guarantee holds regardless of the specific choice of the initial iterate, we observe that if we change the initial iterate, we need to tune the parameters in the proposed algorithm again to achieve a fast convergence speed.
We are working on incorporating line search techniques into the proposed algorithm, in order to accelerate the proposed algorithm and increase its robustness in convergence speed with respect to the choice of the initial iterate.

\bibliographystyle{alpha}
\bibliography{ref,list}

\appendix
%


\section{Dimension Dependence}
\label{Dim dependence}
It is well known that entropic mirror descent is efficient when the parameter dimension is large. 
In particular, assuming the objective function is convex and Lipschitz, the convergence speed of entropic mirror descent is only logartihmically dependent on the parameter dimension, exhibiting a polynomial improvement over gradient descent (see, e.g., \cite{Juditsky2012,Bubeck2015a}). 
Unfortunately, the objective functions we consider are not Lipschitz. 
We conduct numerical experiments to study how the convergence speed of entropic mirror descent with the Polyak step size scales with the parameter dimension. 

We consider the effect of the dimension of $\mathcal{H}_B$ on the iteration complexity, the number of iterations required to achieve an optimization error requirement, of entropic mirror descent with the Polyak step size. 
The iteration number and the corresponding elapsed time are recorded when the error $f(x_t)-f(x^\star)$ is less than $10^{-5}$. 
The results for minimizing the Petz--Augustin information \eqref{eq:Petz_Augustin} and sandwiched Augustin information \eqref{eq:sand_Augustin} are presented in Figure \ref{Iter_Complexity_Petz_SandA}. 
For computing the Petz--Augustin information, we set $\delta_1 = 4,$ $\delta = 10^{-5},$ $\rho = 1.25,$ and $\beta = 0.75$ in our method; 
for computing the sandwiched Augustin information, we set $\delta_1 = 5,$ $\delta = 10^{-5},$ $\rho = 1.1,$ and $\beta = 0.9$.
The support cardinality of $X$ is 10. 
The results for minimizing the conditional sandwiched R\'enyi entropy \eqref{Conditional sandwiched Renyi Entropy} and sandwiched R\'enyi information \eqref{Sandwiched Renyi Information} are presented in Figure \ref{Iter_Complexity_Cond_SandR}. 
For both cases, we set $\delta_1 = 1,$ $\delta = 10^{-5},$ $\rho = 1.1,$ and $\beta = 0.9$ in our method. 
The dimension of $\mathcal{H}_A$ is also 10.
We can observe that the iteration complexities grow sub-linearly with the dimension. 
We also plot the elapsed time for the reader's reference. 
\begin{figure}[ht]
    \centering
     \includegraphics[width = 7.7cm]{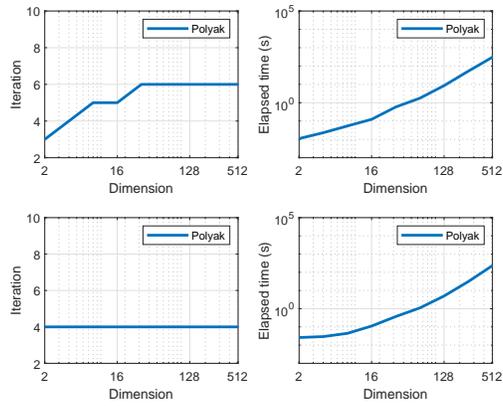}
     \caption{Iteration complexities of computing the Petz--Augustin information (upper) and sandwiched Augustin information (bottom), both of order \( \alpha = 0.5 \).
     }
     \label{Iter_Complexity_Petz_SandA}
\end{figure}

\begin{figure}[ht]
    \centering
     \includegraphics[width = 7.7cm]{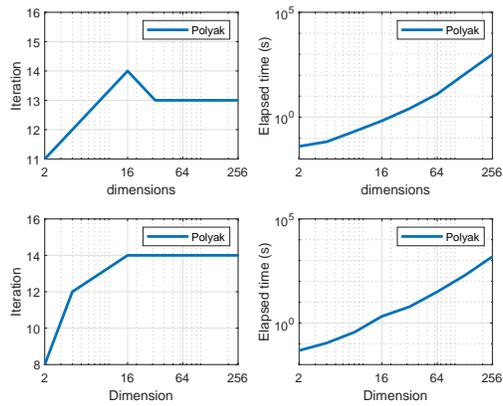}
     \caption{Iteration complexities of computing the conditional sandwiched R\'enyi entropy (upper) and sandwiched R\'enyi information (bottom), both of order \( \alpha = 10 \). 
     }
     \label{Iter_Complexity_Cond_SandR}
\end{figure}

\end{document}